\documentclass[11pt]{article}
\usepackage[a4paper, top=1in, bottom=1in, left=1in, right=1in]{geometry}
\usepackage{array,xspace,multirow,hhline,tikz,colortbl,tabularx,booktabs,fixltx2e,amsmath,amssymb,amsfonts,amsthm}
\usepackage{enumitem}
\usepackage{times}
\usepackage{soul}
\usepackage{url}
\usepackage[hidelinks]{hyperref}
\usepackage[utf8]{inputenc}
\usepackage[small]{caption}
\usepackage{graphicx}
\usepackage{amsmath}
\usepackage{booktabs}
\usepackage{algorithm}
\usepackage{algorithmic}
\usepackage[switch]{lineno}
\usepackage{natbib}

\usepackage{amssymb}
\usepackage{xcolor}
\usepackage{balance}

\newtheorem{theorem}{Theorem}

\newtheorem{claim}{Claim}

\newtheorem{definition}{Definition}
\newtheorem{observation}{Observation}

\newtheorem{lemma}{Lemma}

\newtheorem{proposition}{Proposition}

\newcommand{\cI}{\mathcal{I}}

\newcommand{\bZ}{\mathbb{Z}}

\newcommand{\cN}{\mathcal{N}}

\newcommand{\bv}{\mathbf{v}}
\newcommand{\bc}{\mathbf{c}}
\newcommand{\bs}{\mathbf{s}}

\date{}
\title{Bin Packing and Covering: Pushing the Frontier on the Maximin Share Fairness}

\author{Bo Li\thanks{Department of Computing, The Hong Kong Polytechnic University, China. comp-bo.li@polyu.edu.hk} \and
Ankang Sun\thanks{Department of Computing, The Hong Kong Polytechnic University, China ankang.sun@polyu.edu.hk} \and  Zunyu Wang\thanks{Department of Computing, The Hong Kong Polytechnic University, China. zunyu1.wang@connect.polyu.hk} \and
Yu Zhou\thanks{Institute of AI and Future Networks, Beijing Normal University, China. yu.zhou@bnu.edu.cn}}

\begin{document}
\maketitle

\begin{abstract}
We study a fundamental fair allocation problem, where the agent's value is determined by the number of bins either used to pack or cover the items allocated to them. 
Fairness is evaluated using the maximin share (MMS) criterion.
This problem is not only motivated by practical applications, but also serves as a natural framework for studying group fairness.
As MMS is not always satisfiable, we consider two types of approximations: cardinal and ordinal. 
For cardinal approximation, we relax the requirements of being packed or covered for a bin, and for ordinal approximation, we relax the number of bins that are packed or covered.
For all models of interest, we provide constant approximation algorithms. 
\end{abstract}

\section{Introduction}

Fair allocation of indivisible items has recently attracted considerable interest in the fields of computational economics \citep{DBLP:journals/ior/AkramiACGMM25}, artificial intelligence \citep{DBLP:journals/ai/SeddighinS24}, and theoretical computer science \citep{DBLP:conf/stoc/GargMQ25}. 
Most theoretical work focuses on broad valuation classes, including additive, submodular, and subadditive valuations \citep{DBLP:conf/soda/AkramiG24,seddighin2025beating}.
Motivated by real-world applications, there is a recent research trend on how to integrate the combinatorial structure of items into the design of fair algorithms.
Various models have been investigated, such as the order delivery problem \citep{DBLP:conf/ijcai/HosseiniNW24}, budget-feasible constraints \citep{DBLP:conf/sigecom/Barman0SS23}, and scheduling constraints \citep{DBLP:journals/corr/abs-2402-04353}.

In this paper,  we revisit the model introduced in \citep{DBLP:conf/nips/0037WZ23}, where items are allocated to agents who have bin packing costs \citep{DBLP:journals/dam/MartelloT90,coffman2013bin}. 
This model captures a broad range of real-world applications.
For example, when cloud providers allocate computing resources (RAM, CPU, storage) to different users or tenants.
Users want to minimize the number of virtual machines used to run their tasks.
A similar situation arises in logistics, where parcels are distributed among delivery agents. Each agent seeks to minimize the number of trucks needed to complete their deliveries. 
Formally, the problem can be modeled as allocating $m$ indivisible chores (interpreted as items) among $n$ agents, where the size of each item is agent-specific. 
Upon receiving a bundle of items, each agent must pack them into bins such that the total size in any bin does not exceed a given capacity. The cost incurred by an agent is defined as the minimum number of bins required to pack all her allocated items feasibly.

\citet{DBLP:conf/nips/0037WZ23} considered the maximin share (MMS) fairness introduced by \citet{budish2011combinatorial}.
The criterion is defined through a thought experiment: each agent divides all items into 
$n$ bundles, but only receives the bundle with the largest cost, where $n$ is the number of agents. The MMS of the agent is the smallest cost they can guarantee themselves in this experiment.
For the bin packing model, the MMS is essentially the minimum $\kappa$ such that all items can be partitioned into $n$ bundles so that the items in every bundle can be packed into $\kappa$ bins.
An allocation is regarded as MMS fair if every agent's cost is at most her MMS.
However, \citet{DBLP:conf/nips/0037WZ23} demonstrate that MMS fair allocations may not always exist. As a result, they study an ordinal relaxation of MMS fairness. 
An allocation is said to be $\alpha$-ordinal-approximate MMS fair if the number of bins each agent uses to pack the assigned items does not exceed $\alpha$ times her MMS.
They show that a simple round-robin algorithm achieves a 2-ordinal approximation, and this ratio is tight. Nevertheless, they observe that allowing a small constant additive loss enables an improved multiplicative approximation ratio of $\frac{3}{2}$. 
In this work, we further improve the asymptotic ordinal approximation ratio.
Furthermore, we propose to study the {\em cardinal} approximation. 
Informally, an allocation is $\alpha$-cardinal-approximate MMS fair if each agent can pack the assigned items using MMS bins, with each bin exceeding its capacity by at most a factor of $\alpha$. 

While the bin packing model has been studied, the dual problem -- bin covering model where items are goods \citep{DBLP:journals/jal/AssmannJKL84} -- has not been explored. 
That is, upon receiving a set of items, the agent uses them to cover bins, where a bin is considered covered if the total size of items in it meets or exceeds the bin's capacity. 
Here, the agent's utility is the maximum number of bins they are able to cover. 
This setting also widely appears in the real world.
For example, consider the medical kit assembly problem, where hospitals are allocated medical supplies to assemble emergency or surgical kits. A kit is considered usable only when it includes all essential items meeting or exceeding a required threshold (e.g., a complete set of PPE).
Similarly, in IT hardware deployment problem, an IT department receives various hardware components (CPUs, RAM, drives) and wants to assemble as many fully functional computers (bins) as possible, each meeting a minimum specification.
For the bin covering model, the MMS of each agent is the maximum $\kappa$ such that all items can be partitioned into $n$ bundles, with the items in each bundle capable of covering $\kappa$ bins.
Accordingly, an MMS fair allocation requires every agent's utility to be no smaller than MMS.
In this work, for the first time, we study the ordinal and cardinal approximations of MMS fairness for the bin covering model.

The bin packing and covering models not only reflect practical applications with such combinatorial structures, but also provide a natural model to study group fairness by treating agents as groups and bins as individual group members. 
The cardinal approximation quantifies how much each member may be overburdened or undercompensated, while the ordinal approximation seeks to minimize the number of additional members required -- ensuring no member is over-utilized -- or to maximize the number of fully compensated members.

\subsection{Main Contribution}

\begin{table}[]
    \centering
    \renewcommand{\arraystretch}{1.2}
    \begin{tabular}{c|cc}
        Setting & Ordinal Approximation & Cardinal Approximation \\
        \hline
        Bin Covering & $\frac{3}{4}$ (Theorem \ref{thm:cover:ordinal})  & $\frac{2}{3}$ (Theorem \ref{thm:cover:cardinal})\\
        Bin Packing & $\frac{4}{3}$ (Theorem \ref{thm::4/3-chore}) & $\frac{4}{3}$ \citep{DBLP:conf/icml/LiWX}
    \end{tabular}
    \caption{Main results:
    For the results on ordinal approximations, we admit a small constant additive loss. 
    }
    \label{tab:main-results}
\end{table}

In this paper, we study the extent to which MMS fairness can be satisfied in the bin covering and packing models. 
The MMS criterion is relaxed in two dimensions, either cardinal or ordinal.
Our algorithmic results are summarized in Table \ref{tab:main-results}.

\paragraph{Bin Covering Model}
We start with the bin covering model, where the items are goods and bring non-negative utilities to the agents.
We first observe that exact MMS cannot be satisfied even when there are two agents. 
That is, there are instances with two agents for which no matter how items are allocated, there is at least one agent whose items can cover strictly fewer bins than their MMS.
This result shows a sharp contrast to the traditional additive setting, for which an MMS allocation always exists for two agents. 
We then focus on the approximations.

We first design an algorithm that guarantees $\frac{2}{3}$-cardinal-approximate MMS fair: it always computes an allocation where each agent can use their assigned items to cover their MMS number of bins, with the total size in each bin reaching at least $\frac{2}{3}$ of its capacity. 
We propose an algorithm that, at each step, selects an agent to partition the remaining items into subsets, with the number of subsets equal to the number of remaining agents.
Then, the algorithm matches the subsets to the agents in a way that each unmatched agent is not satisfied with any of the matched subsets.
This idea has been considered in the literature \citep{DBLP:journals/isci/Aigner-HorevS22}, but our algorithm differs in the following two aspects: (1) instead of arbitrarily selecting an agent, we choose a specific remaining agent according to a certain rule, and (2) our algorithm imposes constraints on the subsets partitioned by the selected agent.
These adaptions are necessary as simply applying the existing algorithm does not a guarantee $\frac{2}{3}$-cardinal-approximate MMS fair solution.

Next, we consider the ordinal approximation, but find that it is impossible to guarantee any non-zero approximation in general. 
This limitation occurs in cases where each agent's MMS is 1, yet in every possible allocation, there exists at least one agent whose items cannot cover even a single bin. 
However, as the MMS value increases, the approximation improves significantly. 
This observation leads us to consider allowing a small constant additive loss in the approximation. 
Based on this, we design an algorithm that guarantees each agent receives utility of at least
$\frac{3}{4}\kappa_i - \frac{7}{4}$, where $\kappa_i $ is agent $a_i$'s MMS.

\paragraph{Bin Packing Model}

For the bin packing model, the items can be viewed as chores and bring cost to the agents.
For the cardinal approximation, we observe an equivalent relationship with the \textit{homogeneous} job scheduling model studied in \citep{DBLP:conf/icml/LiWX}.
\footnote{Given any bin packing instance, 
we can transfer it into a homogeneous job scheduling instance by setting the number of machines owned by each group to the MMS value of the corresponding agent while keeping the item sizes unchanged. 
One can see that the MMS value (i.e., the makespan) of each group in the job scheduling instance is equal to the capacity of the bins owned by the agent. 
Accordingly, an $\alpha$-GMMS allocation in the job scheduling instance leads to an $\alpha$-CMMS allocation in the bin packing instance.
For the other direction, we can transfer any homogeneous job scheduling instance to a bin packing one by setting the MMS value of each group to the capacity of the bins owned by the corresponding agent.}
Thus, the results in \citep{DBLP:conf/icml/LiWX} implies a $\frac{4}{3}$-cardinal-approximation.
For the ordinal approximation, we prove that there always exists an allocation in which each agent $a_i$ has a cost no greater than $\frac{4}{3}\kappa_i+\frac{4}{3}$, where $\kappa_i$ is the MMS of $a_i$.
For any $\kappa_i \geq 2$, our result improves upon that in \citet{DBLP:conf/nips/0037WZ23} which showed the existence of an allocation where each $a_i$ has cost at most $\frac{3}{2}\kappa_i+1$.
This improvement relies on a different classification of the items and a more refined analysis on the structure of the MMS partition of agents.

\subsection{Other Related Work}
\label{ap:related_work}
The problem of fair division has been well studied since 1948 when \citet{steihaus1948problem} first formalized the fairness notion of proportionality (PROP), where each agent receives at least $\frac{1}{n}$ of his total value for all items.
Another popular fairness notion is envy-freeness (EF) \citep{foley1967resource}, where no agent wants to exchange her bundle with any other agent's.
For indivisible items, PROP and EF are hard to satisfy, which gives rise to their relaxations, such as {envy-freeness up to one item} (EF1) \citep{DBLP:conf/sigecom/LiptonMMS04,budish2011combinatorial}, {envy-freeness up to any item}
(EFX) \citep{DBLP:conf/ecai/GourvesMT14,DBLP:journals/teco/CaragiannisKMPS19}, and {maximin share} (MMS) \citep{budish2011combinatorial}.
EF1 allocations always exist for both goods and chores with monotone valuations \citep{DBLP:conf/sigecom/LiptonMMS04,DBLP:conf/approx/BhaskarSV21}.
MMS allocations may not exist even for additive valuations \citep{DBLP:conf/aaai/KurokawaPW16,DBLP:journals/jacm/KurokawaPW18} and the existence of EFX allocations is still unknown.
Accordingly, the approximations of MMS and EFX are extensively studied by follow-up works, e.g., \citep{DBLP:conf/soda/AkramiG24,DBLP:conf/sigecom/GhodsiHSSY18,DBLP:journals/teco/BarmanK20} for goods and \citep{DBLP:journals/ai/ZhouW24,DBLP:conf/sigecom/HuangS23,DBLP:conf/sigecom/HuangL21,DBLP:journals/corr/abs-2407-03318} for chores.
Besides additive valuations, the approximation of MMS with submodular, XoS and subadditive valuations has been studied in \citep{DBLP:journals/teco/BarmanK20,DBLP:conf/sigecom/GhodsiHSSY18,DBLP:journals/corr/abs-2303-12444,DBLP:conf/nips/AkramiMSS23,DBLP:journals/ai/SeddighinS24} for goods and in \citep{DBLP:conf/nips/0037WZ23} for chores.

\section{Preliminaries} \label{sec:firstpage}
For any integer $k \in \bZ^+$, denote $[k]=\{1,\ldots,k\}$.
We consider the problem of allocating $m$ indivisible items 
$M=\{e_1,\ldots,e_m\}$ among $n$ agents $N= \{ a_1,\ldots,a_n\}$.
The agents have subjective opinions on the sizes of items, and the size of item $e_j$ to agent $a_i$ is denoted as $s_{i,j}$.
Without loss of generality, assume that the bin capacity of every agent is 1 and $0<s_{i,j}\le 1$ for all $i,j$.
For simplicity, let $s_i(S) = \sum_{e_j \in S}s_{i,j}$ and $s_i(e_j) =s_{i,j}$.
Let $\bs = (s_1,\ldots,s_n)$.
A fair allocation instance is denoted by $\cI = (N, M, \bs)$.
An allocation, denoted by $\mathbf{A}=(A_1,\ldots, A_n)$, is an $n$-partition of $M$ where $A_i$ is the items allocated to agent $a_i$ such that $A_1\cup\cdots\cup A_n = M$ and $A_i\cap A_j = \emptyset$ for all $i\neq j$.
When items are goods, agents have bin covering valuations $v_i:2^M\to \mathbb{Z}_{\ge 0}$, where $v_i(S)$ is the maximum number of bins that can be covered by items in $S \subseteq M$, i.e., the maximum $\tau\ge 0$ such that there is a $\tau$-partition $(S_1, \ldots, S_{\tau})$ such that 
$s_i(S_j) \ge 1$ 
for all $j$.
When items are chores, agents have bin packing cost functions $c_i:2^M\to \mathbb{Z}_{\ge 0}$, where $c_i(S)$ is the minimum number of bins that can pack all items in $S$, i.e., the minimum $\tau\ge 0$ such that there is a $\tau$-partition $(S_1, \ldots, S_{\tau})$ such that 
$s_i(S_j) \le 1$
for all $j$.
Denote by $\bv=(v_1,\ldots,v_n)$ and $\bc=(c_1,\ldots,c_n)$.

We consider the maximin share (MMS) fairness.
For any positive integer $k$ and set of items $S$, denote by $\Pi(S,k)$ the set of all $k$-partitions of $S$.
For the bin covering model, agent $a_i$'s MMS is defined as
\[
\kappa_i = \max_{(X_1,\ldots,X_n)\in \Pi(M,n)}\min_{j\in [n]} v_i(X_j).
\]
That is, $\kappa_i$ is the maximum $\kappa$ such that $a_i$ can partition all items into $n$ bundles and every bundle can cover no smaller than $\kappa$ bins.
A partition that reaches $\kappa_i$ is called an MMS partition of $a_i$.
An allocation $(A_1, \ldots, A_n)$ is MMS fair if $v_i(A_i) \ge \kappa_i$ for all agents.
It is worth noting that an MMS fair allocation may not exist even for instances with only two agents, which shows a sharp contrast to the traditional additive setting. 
\citet{DBLP:conf/wine/FeigeST21} designed a hard instance for four agents with additive valuations, where the first two agents have the same valuation and so do the other two agents. 
Every agent can partition all the items into four bundles and each bundle is valued exactly 1 (after normalization). 
However, no matter how the items are allocated, there is at least one agent whose value is strictly smaller than 1.
This example can be adapted to the bin covering setting with two agents, where the values correspond to item sizes and each bin has a capacity of 1.
In this case, the MMS values of both agents are 2, however, there will always be an agent who cannot cover two bins.

Therefore, we consider two ways to relax the requirements of MMS fairness.

\begin{definition}
    Given a bin covering instance and $0\le \alpha \le 1$,
    an allocation $(A_1, \ldots, A_n)$ is $\alpha$-cardinal-approximate MMS fair ($\alpha$-CMMS) if every agent $a_i$ can partition $A_i$ into $\kappa_i$ bundles $(A_{i,1},\ldots,A_{i,\kappa_i})$ such that 
    $s_i(A_{i,j}) \ge \alpha$
    for all $j \in [\kappa_i]$.
\end{definition}

\begin{definition}
\label{def:goods:ordinal}
    Given a bin covering instance and $0\le \alpha \le 1$,
    an allocation $(A_1, \ldots, A_n)$ is $\alpha$-ordinal-approximate MMS fair ($\alpha$-OMMS) if $v_i(A_i) \ge \alpha\kappa_i - c$ for some constant $c\ge 0$.
\end{definition}

Note that the constant additive loss $-c$ is important in the above definition. 
If $c = 0$, we are not able to have $\alpha>0$, even when $\kappa_i = 1$ and $n=3$.
We can adapt the hard instance with three agents designed by 
\citet{DBLP:conf/wine/FeigeST21} to the bin covering setting with three agents like above. 
Here, there will always be an agent whose items cannot pack a bin, resulting in $\alpha$ being 0 in Definition \ref{def:goods:ordinal}.
Thus, for ordinal approximation, we focus solely on the asymptotic approximation ratio.

\medskip

For the bin packing model, agent $a_i$'s MMS is defined as
\[
\kappa_i = \min_{(X_1,\ldots,X_n)\in \Pi(M,n)}\max_{j\in [n]} c_i(X_j).
\]
That is, $\kappa_i$ is the minimum $\kappa$ such that $a_i$ can partition all items into $n$ bundles and the items in every bundle can be packed into $\kappa$ bins.
Similarly, a partition that reaches $\kappa_i$ is called an MMS partition of $a_i$.
An allocation $(A_1, \ldots, A_n)$ is MMS fair if $c_i(A_i) \le \kappa_i$ for all agents, and cardinal and ordinal approximations are considered to relax the requirements.

\begin{definition}
    Given a bin packing instance and $\alpha \ge 1$,
    an allocation $(A_1, \ldots, A_n)$ is $\alpha$-cardinal-approximate MMS fair ($\alpha$-CMMS) if every agent $a_i$ can partition $A_i$ into $\kappa_i$ bundles $(A_{i,1},\ldots,A_{i,\kappa_i})$ such that 
    $s_i(A_{i,j}) \le \alpha$
    for all $j \in [\kappa_i]$.
\end{definition}

\begin{definition}
\label{def:chores:ordinal}
    Given a bin packing instance and $\alpha \ge 1$,
    an allocation $(A_1, \ldots, A_n)$ is $\alpha$-ordinal-approximate MMS fair ($\alpha$-OMMS) if $c_i(A_i) \le \alpha\kappa_i + c$ for some constant $c\ge 0$.
\end{definition}

Reading the ordinal approximation, if $c = 0$, it is proved in \citep{DBLP:conf/nips/0037WZ23} that the tight approximation ratio is 2.
But if we allow a small constant additive loss, the multiplicative ratio can be improved to $\frac{3}{2}$.
In our paper, we continue studying the asymptotic bound and improve this ratio to $\frac{4}{3}$.

\medskip

Finally, we call an instance identical ordering (IDO) if all agents agree on an ordering of the items regarding their sizes, i.e., $s_{i,1}\ge s_{i,2} \ge \cdots \ge s_{i,m}$ for all $a_i$.
It is widely known that IDO instances are the hardest to approximate MMS for the instances with additive valuations \citep{DBLP:journals/teco/BarmanK20,DBLP:conf/sigecom/HuangL21}.
This is also true for bin covering and bin packing models, and the proof is almost the same as that in these references.
Thus we state the following and omit the proof.

\begin{lemma}
    If there is an algorithm that ensures $\alpha$-CMMS (or $\alpha$-OMMS) allocations for all IDO instances, there is another algorithm that ensures $\alpha$-CMMS (or $\alpha$-OMMS)  allocations for arbitrary instances.
    Further, if the former algorithm runs in polynomial time, the latter does as well.
\end{lemma}

In all the subsequent sections, we only consider IDO instances. 

\section{Cardinal Approximation for the Bin Covering Model}
\label{sec:cover:cardinal}

In this section, we present the algorithm to compute $\frac{2}{3}$-CMMS allocations for the bin covering model.
That is, given each agent $a_i$'s MMS value $\kappa_i$ -- the maximum number of bins each bundle can cover in the MMS partition -- we compute an allocation in which the items received by each agent $a_i$ can be partitioned into $\kappa_i$ bundles, each with a total size of at least $\frac{2}{3}$.

\begin{theorem}
\label{thm:cover:cardinal}
For any bin covering instance, a $\frac{2}{3}$-CMMS allocation exists. 
\end{theorem}

To simplify the analysis, we assume in the MMS partition of each agent $a_i$, every bin is exactly covered. 
In other words, $M$ can be partitioned into $n\kappa_i$ bundles $X_1,\ldots,X_{\kappa_i\cdot n}$ such that $s_i(X_j)=1$ for all $j \in [\kappa_i\cdot n]$.
This assumption is made without loss of generality: we can always rescale the item sizes so that for any $e_l\in M$, if $e_l\in X_j$, $s'_{i,l}=\frac{s_{i,l}}{s_i(X_j)}$.
Then, if we can achieve $\frac{2}{3}$-CMMS in the scaled instance,  the same guarantee holds for the original instance.

\subsection{The Algorithm}

For each agent $a_i$, item $e$ is large if $\frac{2}{3}\leq s_i(e) \leq 1$, and is medium if $\frac{1}{3}< s_i(e) <\frac{2}{3}$. 
If an item is neither large nor medium, it is small.
As the capacity of the bin is 1, for each bundle in the MMS partition for $a_i$, aside from small items, it may contain at most (1) a large item, (2) a medium item, and (3) two medium items.
We now adjust the MMS partition as follows: if there exists medium items $e,e'$ such that $s_i(e)>s_i(e')$, and in the partition, $e$ is packed in a bin with another medium item while $e'$ is the only medium item in the bin, then swap the positions of $e$ and $e'$.
Repeat the process until no two medium items satisfying the swap condition remains, resulting in the \emph{refined-partition}.

For each bundle in the refined-partition, again, aside from small items, it may contain at most (1) a large item, (2) a medium item, and (3) two medium items.
Let $F$ denote the set of large and medium items that, in the refined-partition, are placed solely in a bundle without any other medium items. 
Let $H$ denote the set of medium items that, in the refined-partition, are placed in a bundle with another medium item. Thus, $F\cup H$ equals the set of large and medium items. 
For each $a_i$, as there are $\kappa_i\cdot n$ bundles in the refined-partition and items in $H$ are placed in pairs, it holds that
$
|F|+\frac{|H|}{2} \leq \kappa_i \cdot n.
$
The sets $F$ and $H$ can differ among agents. We abuse notation and use $F$ and $H$ for every agent, and will clarify the underlying agent from the context when needed.

Next, we introduce the high-level idea of the algorithm. At each step, we select an agent $a_i$ from the remaining agents, denoted as $\mathcal{N}$. Agent $a_i$ partitions the remaining items into $|\mathcal{N}|$ parts such that she is happy with every part.
In our context, this translates to that every part $S$ can be partitioned into $\kappa_i$ bundles, each with a total item size of at least $\frac{2}{3}$ (when targeting $\frac{2}{3}$-CMMS).
Then create a bipartite graph with remaining agents on one side and the constructed $|\mathcal{N}|$ parts on the other side. There is an edge between an agent and a part if the agent is happy when receiving that part.
\emph{An envy-free matching with regards to the agents} is a matching in which every unmatched agent is not incident to any matched part. \citet{DBLP:journals/isci/Aigner-HorevS22} proved that an envy-free matching with regards the agents always exists once the selected group can create $|\mathcal{N}|$ satisfied parts. Moreover, they provided a polynomial time algorithm to compute the maximin cardinality envy-free matching. 
Match agents to parts based on the envy-free matching, which guarantees that for every pair of unmatched $a_j$ and a matched part $S$, agent $a_j$ is not happy with receiving $S$.

Similar algorithmic ideas have been applied in computing approximate MMS allocations in other contexts \citep{DBLP:journals/corr/abs-2404-11582,DBLP:conf/aaai/AkramiR25a}. 
Our techniques differ from those existing algorithms in the following two aspects: at each step, (1) instead of choosing an arbitrary agent, we select the agent with the maximum total number of large and medium items, 
and (2) instead of allowing the selected agent to create $|\mathcal{N}|$ satisfied parts in an arbitrary way, we require that these $|\mathcal{N}|$ parts must partition large and medium items in a balanced manner.
In specific, the selected agent creates the parts based on a (predetermined) $|\mathcal{N}|$-\emph{arrangement}, defined as follows: according to the non-descending order of item size, place remaining items into $|\mathcal{N}|$ columns such that $j$-th column contains $j$-th, $(2|\mathcal{N}|-j+1)$-th, $(2|\mathcal{N}|+j)$-th,... largest size item.
When the selected agent creates $|\mathcal{N}|$ parts, each bundle must contain every large and medium item placed in some column in the $|\mathcal{N}|$-arrangement.
The proposed algorithm is formally introduced as Algorithm \ref{alg:2/3-cardinal-goods}.

\begin{algorithm}[ht!]
        \caption{ Computing cardinal approximate MMS allocations}
	\label{alg:2/3-cardinal-goods}
	\begin{algorithmic}[1]
        \REQUIRE An instance $(N,M,\bs)$ and $\kappa_i$ for every $a_i \in N$.
		\ENSURE A $\frac{2}{3}$-CMMS allocation $\mathbf{A} = (A_1,\ldots,A_n)$.
        \STATE Create the $|N|$-arrangement and initialize $\cN \gets N$.
        \WHILE{$\cN\neq \emptyset$}
        \STATE Select $a_i\in \cN$ with the maximum total number of large and medium items, breaking ties arbitrarily.\label{step:alg-2/3-select-agents}
        \STATE Create $|\cN|$ pairwise disjoint parts $S_1,\ldots,S_{|\cN|}$ as follows. For each $S_j$, $\forall j \leq |\cN|$, first distribute all large and medium items in the $j$-th column of the $|\cN|$-arrangement into at most $\kappa_i$ bundles such that the total size of items in each bundle is at most 1 for $a_i$.
        Then add small items to these $\kappa_i$ bundles until the total size of items in each bundle is at least $\frac{2}{3}$ for $a_i$.
        \label{step:alg-2/3-properties} 
        \STATE Create a bipartite graph with vertices $\cN\cup (S_1,\ldots, S_{|\cN|})$. Add edge $(a_\ell,S_j)$ if and only if $S_j$ can be partitioned into $\kappa_\ell$ bundles, each with total item size at least $\frac{2}{3}$ for $a_\ell$.
        Compute a maximum-cardinality envy-free matching and assign each matched part to its matched agent.\label{step:alg-2/3-matching}
        \STATE Update $\cN$ and $M$ by removing matched agents and items. 
        Then update the $|\cN|$-arrangement.\label{step:alg-2/3-update}
        \ENDWHILE
	\end{algorithmic}
\end{algorithm}

\subsection{The Proof of Theorem \ref{thm:cover:cardinal}}
We now prove that Algorithm \ref{alg:2/3-cardinal-goods} can return an allocation. It suffices to show that in each round, the selected agent can create the $|\mathcal{N}|$ parts described in Line \ref{step:alg-2/3-properties}. First, consider the first round of the algorithm.

\begin{lemma}\label{lem::cardinal-goods-l-1}
    At the first round of the while-loop, for each agent $a_i$, she can create $n$ pairwise disjoint parts $P_1,\ldots,P_n$ following Line \ref{step:alg-2/3-properties} such that for any $j\in [n]$, $P_j$ can be partitioned into $\kappa_i$ bundles, each with a total item size of at least $\frac{2}{3}$.
\end{lemma}
\begin{proof}
    For any $j\in [n]$, let $F_j\subseteq F$ (resp., $H_j\subseteq H$) denote the items in $F$ (resp., $H$) in the $j$-th column of the $n$-arrangement. We focus on the $j$-th column and first show that $a_i$ can place all items in $F_j\cup H_j$ into at most $\kappa_i$ bundles (or bins), each having at most two items with total size at most 1 for $a_i$.
    In specific, we always place each item in $F_j$ into a single bin, and thus, it suffices to show that items in $H_j$ can be placed into at most $\kappa_i - |F_j|$ bundles, each having at most two items with total size at most 1.

    Suppose $ \lfloor \frac{|F|}{2n} \rfloor = x$ with $x\in \mathbb{N}^+\cup\{0\}$ and $H=\{h_1,\ldots,h_{|H|}\}$ with $s_i(h_1)\geq s_i(h_2) \geq \ldots \geq s_i(h_{|H|})$. For simplicity, instead of $s_i(\{e\})$, we write $s_i(e)$.
    In the MMS refined-partition of $a_i$, items in $H$ form $\frac{|H|}{2}$ pairs, satisfying that the two items in each pair have total size at most 1. Then it must hold that $s_i(h_1\cup h_{|H|}) \leq 1$, as otherwise, the total size of $h_1$ and any item in $H$ exceeds 1, contradicting the refined-partition.
    By the same reasoning, $s_i(h_2\cup h_{|H|-1}) \leq 1$ holds, and indeed, for any $r,r'$ with $r+r'\geq |H|+1$, $s_i(h_r\cup h_{r'}) \leq 1$.
    \begin{observation}\label{obs::cardinal-goods-H}
        For any $r,r'$ with $r+r'\geq |H|+1$, $s_i(h_r\cup h_{r'}) \leq 1$.
    \end{observation}
    Next, we below split the proof based on the values of $|F|$ and $j$.
    \smallskip
    
    \noindent\textbf{\underline{Case 1:}} $|F|=2nx+y$ with $y\in \{1,\ldots,n\}$, and $j\in \{1,\ldots,y\}$.
    In this case, we have $|F_j|=2x+1$. Then the remaining is to show that one can always place items in $H_j$ into at most $\kappa_i-2x-1$ bundles, each having at most two items with total size at most 1.
    Note that $|F|=2nx+y\leq \kappa_in$ holds, and hence, $\kappa_i \geq 2x+1$, as $\kappa_i,x,y$ are integers and $y\geq 1$. 
    That $|H_j|=0$ is a trivial case, and then we focus on $|H_j|\geq 1$. By $|F|=2nx+y$ and $|F|+\frac{|H|}{2} \leq \kappa_in$, it holds that $|H|\leq 2\kappa_in-4nx-2y$, which implies $|H_j|\leq 2\kappa_i-4x-2$.
    
    First, consider the case when $|H_j| = 2\kappa_i-4x-2$. Since $|H_j|\geq 1$ and $\kappa_i,x$ are integers, we have $\kappa_i \geq 2x+2$. 
    We enumerate items in $H_j$ as follows:
    $$
    H_j=\left\{ h_r\mid r= 2np-y-j+1 \textnormal{ and } r = 2np-y+j \textnormal{ for } p=1,\ldots, \kappa_i-2x-1\right\}.
    $$
    Due to Observation \ref{obs::cardinal-goods-H} and $|H|\leq 2\kappa_in-4nx-2y$, for each $s\in\{1,\ldots,\frac{|H_j|}{2}\}$, the total size of the $s$-th largest size item and the $(|H_j|+1-s)$-th largest size item in $H_j$ is at most 1. 
    Therefore, one can place items in $H_j$ into $\frac{|H_j|}{2} = \kappa_i-2x-1$ bundles, each having two items with total size at most 1.

    When $|H_j|=2\kappa_i-4x-3$, it must hold that $|H|\leq 2n(\kappa_i-2x-1)+j-y-1$. In the above enumeration, item $h_r$ with $r=2n(\kappa_i-2x-1)-y+j$ no longer exists. If $|H_j|\leq 2$, we have $\kappa_i\geq 2x+1 + |H_j|$, as $\kappa_i$ is an integer.
    Accordingly, $|H_j| \leq \kappa_i-|F_j|$, and we are done. Focus on $|H_j|\geq 3$.
    By Observation \ref{obs::cardinal-goods-H}, for each $s\in\{1,\ldots,\lfloor\frac{|H_j|}{2}\rfloor\}$, the total size of the $(s+1)$-th largest size item and the $(|H_j|+1-s)$-th largest size item in $H_j$ is at most 1. 
    Then, excluding the largest value item in $H_j$, the remaining items can be placed into $\kappa_i-2x-2$ bundles, each having two items with total size at most 1. Therefore, $a_i$ can place all items in $H_j$ into $\kappa_i-2x-1$ bundles, each having at most two items with total size at most 1.

    When $|H_j|\leq 2\kappa_i-4x-4$, it must hold that $|H|\leq 2n(\kappa_i-2x-2)+2n-j-y$. As $\kappa_i \geq 2x+2+\frac{|H_j|}{2}$ and $\kappa_i$ is an integer, we have $\kappa_i \geq 2x+1 + |H_j| = |F_j|+|H_j|$ when $|H_j|\leq 3$.
    We can focus on $|H_j|\geq 4$. By Observation \ref{obs::cardinal-goods-H}, for each $s\in\{1,\ldots,\lfloor\frac{|H_j|}{2}\rfloor-1\}$, the total size of the $(s+2)$-th largest value item and the $(|H_j|+1-s)$-th largest size item in $H_j$ is at most 1. 
    Thus, aside from the two largest size items in $H_j$, the remaining items can be placed into $\lceil \frac{|H_j|-2}{2} \rceil \leq \kappa_i-2x-3$ bundles, each having at most two items with total size at most one. Therefore, agent $a_i$ can place all items in $H_j$ into at most $\kappa_i-2x-1$ bundles, each having at most two items with total size at most 1.
    \smallskip
    
    \noindent\textbf{\underline{Case 2:}} $|F|=2nx+y$ with $y\in \{1,\ldots,n-1\}$, and $j \in \{y+1,\ldots,n\} $. The case $y=j=n$ is covered in Case 1. 
    In Case 2, we have $|F_j|=2x$.
    Then the remaining is to show that $a_i$ can always place items in $H_j$ into at most $\kappa_i-2x$ bundles, each having at most two items with total size at most 1. 
    Note that $|F|=2nx+y\leq \kappa_in$ holds, and hence, $\kappa_i \geq 2x+1$, as $\kappa_i,x,y$ are integers and $y\geq 1$.
    That $|H_j|=0$ is a trivial case, and then we focus on $|H_j|\geq 1$. By $|F|=2nx+y$ and $|F|+\frac{|H|}{2} \leq \kappa_in$, it holds that $|H|\leq 2\kappa_in-4nx-2y$, which implies $|H_j|\leq 2\kappa_i-4x$ as $j\geq y+1$.
    
    First, consider the case when $|H_j| = 2\kappa_i-4x$. 
    Enumerate items in $H_j$ as follows:
    $$
    H_j = \left\{ h_r\mid r=2n(p-1)+j-y \textnormal{ and } r=2np+1-j-y \textnormal{ for }p=1,\ldots,\kappa_i-2x \right\}.
    $$
    Due to Observation \ref{obs::cardinal-goods-H} and $|H|\leq 2\kappa_in-4nx-2y$, for each $s\in\{1,\ldots,\frac{|H_j|}{2}\}$, the total size of the $s$-th largest size item and the $(|H_j|+1-s)$-th largest size item in $H_j$ is at most 1. 
    Therefore, one can place items in $H_j$ into $\frac{|H_j|}{2} = \kappa_i-2x$ bundles, each having two items with total size at most 1.

    When $|H_j|=2\kappa_i-4x-1$, it must hold that $|H|\leq 2n(\kappa_i-2x)-j-y$. In the above enumeration, item $h_r$ with $r=2n(\kappa_i-2x)+1-j-y$ no longer exists. 
    If $|H_j|\leq 2$, we have $\kappa_i\geq 2x+ |H_j|$, as $\kappa_i$ is an integer. Accordingly, $|H_j|\leq \kappa_i - |F_j|$, and we are done.
    Focus on $|H_j|\geq 3$.
    By Observation \ref{obs::cardinal-goods-H}, for each $s\in\{1,\ldots,\lfloor\frac{|H_j|}{2}\rfloor\}$, the total size of the $(s+1)$-th largest size item and the $(|H_j|+1-s)$-th largest size item in $H_j$ is at most 1. 
    Then, aside from the largest size item in $H_j$, the remaining items can be placed into $\kappa_i-2x-1$ bundles, each having two items with total size at most 1. Therefore, agent $a_i$ can place all items in $H_j$ into $\kappa_i-2x$ bundles, each having at most two items with total size at most 1.

    When $|H_j|\leq 2\kappa_i-4x-2$, it must hold that $|H|\leq 2n(\kappa_i-2x-1)+j-y-1$. If $|H_j|\leq 3$, we have $\kappa_i\geq 2x+|H_j| = |F_j|+|H_j|$, as $\kappa_i$ is an integer. Thus, we can focus on $|H_j|\geq 4$.
    Due to Observation \ref{obs::cardinal-goods-H}, for each $s\in\{1,\ldots,\lfloor\frac{|H_j|}{2}\rfloor-1\}$, the total size of the $(s+2)$-th largest size item and the $(|H_j|+1-s)$-th largest size item in $H_j$ is at most 1. 
    Thus, aside from the two largest size items in $H_j$, the remaining items can be placed into $\lceil \frac{|H_j|-2}{2} \rceil \leq \kappa_i-2x-2$ bundles, each having at most two items with total size at most 1. Therefore, agent $a_i$ can always place all items in $H_j$ into at most $\kappa_i-2x$ bundles, each having at most two items with total size at most 1.
    \smallskip
    
    \noindent\textbf{\underline{Case 3:}} $|F|=2nx+n+y$ with $y\in\{1,\ldots,n\}$, and $j=\{n-y+1,\ldots,n\}$. In this case, we have $|F_j|=2x+2$.
    Then the remaining is to show that one can always place items in $H_j$ into at most $\kappa_i-2x-2$ bundles, each having at most two items with total size at most 1. 
    By $\kappa_in \geq 2nx+n+y$ and $y\geq 1$, we have $\kappa_i \geq 2x+2$ as $\kappa_i$ is an integer.
    That $|H_j|=0$ is a trivial case, and we focus on $|H_j|\geq 1$.
    By $|F|+\frac{|H|}{2} \leq \kappa_in$, it holds that $|H|\leq 2\kappa_in-4nx-2n-2y$, which implies $|H_j|\leq 2\kappa_i-4x-4$.

    First, consider the case when $|H_j|= 2\kappa_i-4x-4$. As $|H_j|\geq 1$, we have $2\kappa_i\geq 4x+5$, implying $\kappa_i\geq 2x+3$ as $\kappa_i$ is an integer.
    Enumerate items in $H_j$ as follows:
    $$
    H_j = \left\{ h_r\mid r=2n(p-1)+n+j-y \textnormal{ and } r=2np+n-j-y+1 \textnormal{ for } p = 1,\ldots, \kappa_i-2x-2 \right\}.
    $$
    Due to Observation \ref{obs::cardinal-goods-H} and $|H|\leq 2\kappa_in-4nx-2n-2y$, for each $s\in\{1,\ldots,\frac{|H_j|}{2}\}$, the total size of the $s$-th largest size item and the $(|H_j|+1-s)$-th largest size item in $H_j$ is at most 1.
    Therefore, one can place items in $H_j$ into $\frac{|H_j|}{2} = \kappa_i-2x-2$ bundles, each having two items with total size at most 1.

    When $|H_j|=2\kappa_i-4x-5$, it must hold that $|H|\leq 2n(\kappa_i-2x-2)+n-j-y$. In the above numeration, item $h_r$ with $r=2n(\kappa_i-2x-2)+n-j-y+1$ no longer exists.
    If $|H_j|\leq 2$, we have $\kappa_i\geq 2x+2+|H_j|=|F_j|+|H_j|$, as $\kappa_i$ is an integer. Thus, we focus on $|H_j|\geq 3$. 
    By Observation \ref{obs::cardinal-goods-H}, for each $s\in\{1,\ldots,\lfloor\frac{|H_j|}{2}\rfloor\}$, the total size of the $(s+1)$-th largest size item and the $(|H_j|+1-s)$-th largest size item in $H_j$ is at most 1. 
    Then, aside from the largest size item in $H_j$, the remaining items can be placed into $\kappa_i-2x-3$ bundles, each having two items with total size at most 1. Therefore, one can place all items in $H_j$ into $\kappa_i-2x-2$ bundles, each having at most two items with total size at most 1.

    When $|H_j|\leq 2\kappa_i-4x-6$, it must hold that $|H|\leq 2n(\kappa_i-2x-3)+n+j-y-1$. If $|H_j|\leq 3$, we have $\kappa_i\geq 2x+2+|H_j|=|F_j|+|H_j|$, as $\kappa_i$ is an integer. Thus, we focus on $|H_j|\geq 4$. 
    Due to Observation \ref{obs::cardinal-goods-H}, for each $s\in\{1,\ldots,\lfloor\frac{|H_j|}{2}\rfloor-1\}$, the total size of the $(s+2)$-th largest size item and the $(|H_j|+1-s)$-th largest size item in $H_j$ is at most 1. 
    Thus, aside from the two largest size items in $H_j$, the remaining items can be placed into $\lceil \frac{|H_j|-2}{2} \rceil \leq \kappa_i-2x-4$ bundles, each having at most two items with total size at most 1. Therefore, agent $a_i$ can always place all items in $H_j$ into at most $\kappa_i-2x-2$ bundles, each having at most two items with total size at most 1.
    \smallskip
    
    \noindent\textbf{\underline{Case 4:}} $|F|=2nx+n+y$ with $y\in\{1,\ldots,n-1\}$, and $j=\{1,\ldots, n-y\}$. In this case, we have $|F_j|=2x+1$. Then the remaining is to show that one can always place items in $H_j$ into at most $\kappa_i-2x-1$ bundles, each having at most two items with total size at most 1. 
    That $|H_j|=0$ is a trivial case, and we focus on $|H_j|\geq 1$.
    By $|F|+\frac{|H|}{2} \leq \kappa_in$, it holds that $|H|\leq 2\kappa_in-4nx-2n-2y$, which implies $|H_j|\leq 2\kappa_i-4x-2$.

    First, consider the case when $|H_j| = 2\kappa_i-4x-2$. Enumerate items in $H_j$ as follows:
    $$
    H_j=\left\{ h_r \mid r=2np+n-j-y+1 \textnormal{ and } r= 2np+n+j-y \textnormal{ for } p=0,\ldots,\kappa_i-2x-2 \right\}.
    $$
    Due to Observation \ref{obs::cardinal-goods-H} and $|H|\leq 2\kappa_in-4nx-2n-2y$, for each $s\in\{1,\ldots,\frac{|H_j|}{2}\}$, the total size of the $s$-th largest size item and the $(|H_j|+1-s)$-th largest size item in $H_j$ is at most 1.
    Therefore, one can place items in $H_j$ into $\frac{|H_j|}{2} = \kappa_i-2x-1$ bundles, each having two items with total size at most 1.

    When $|H_j|=2\kappa_i-4x-3$, it must hold that $|H|\leq 2n(\kappa_i-2x-2)+n+j-y-1$. If $|H_j|\leq 2$, we have $\kappa_i\geq 2x+1+|H_j|=|F_j|+|H_j|$, as $\kappa_i$ is an integer. Thus, we focus on $|H_j|\geq 3$. 
    By Observation \ref{obs::cardinal-goods-H}, for each $s\in\{1,\ldots,\lfloor\frac{|H_j|}{2}\rfloor\}$, the total size of the $(s+1)$-th largest size item and the $(|H_j|+1-s)$-th largest size item in $H_j$ is at most 1. 
    Then, aside from the largest value item in $H_j$, the remaining items can be placed into $\kappa_i-2x-2$ bundles, each having two items with total size at most 1. Therefore, one can place all items in $H_j$ into $\kappa_i-2x-1$ bundles, each having at most two items with total size at most 1.

    When $|H_j|\leq 2\kappa_i-4x-4$, it must hold that $|H|\leq 2n(\kappa_i-2x-2)+n-j-y$.
    If $|H_j|\leq 3$, we have $\kappa_i\geq 2x+1+|H_j|=|F_j|+|H_j|$, as $\kappa_i$ is an integer. Thus, we focus on $|H_j|\geq 4$. 
    Due to Observation \ref{obs::cardinal-goods-H}, for any $s\in\{1,\ldots,\lfloor\frac{|H_j|}{2}\rfloor-1\}$, the total size of the $(s+2)$-th largest size item and the $(|H_j|+1-s)$-th largest size item in $H_j$ is at most 1. 
    Thus, aside from the two largest size items in $H_j$, the remaining items can be placed into $\lceil \frac{|H_j|-2}{2} \rceil \leq \kappa_i-2x-3$ bundles, each having at most two items with total size at most 1. Therefore, agent $a_i$ can always place all items in $H_j$ into at most $\kappa_i-2x-1$ bundles, each having at most two items with total size at most 1. 
\end{proof}

Now suppose that agents $a_1,\ldots,a_{k}$ are matched and each matched agent $a_j$ receives subsets of items $B_j$. In the next round, agent $a_i$ is selected. We prove that agent $a_i$ can still create $n-k$ parts as required in Line \ref{step:alg-2/3-properties}.

\begin{lemma}\label{lem::cardinal-goods-l-2}
    For each $k\in [n-1]$ and any $k$ subsets $B_1,\ldots,B_k$ allocated before $a_i$ is selected, 
    agent $a_i$ can create $n-k$ pairwise disjoint parts $P_1,\ldots,P_{n-k}$ following Line \ref{step:alg-2/3-properties} such that for any $j\in [n-k]$, $P_j$ can be partitioned into $\kappa_i$ bundles, each with a total item size of at least $\frac{2}{3}$.
\end{lemma}
\begin{proof}
    Without loss of generality, assume that each matched $B_j$ contains the large and medium items in $j$-th column of the $n$-arrangement (created in the first round).
    At each round, the selected agent is the one with the maximum total number of large and medium items among all remaining agents.
    Then for each $j\in [k]$, $B_j$ includes $a_i$'s large and medium items in $j$-th column of the $n$-arrangement.
    Moreover for each $j\in [k]$, let $S_j\subseteq B_j$ be the small items for $a_j$, and then each item in $S_1\cup \cdots \cup S_k$ is also small for $a_i$.
    Thus for any $\ell\geq k+1$, the large and medium items (from $a_i$'s perspective) in $\ell$-th column of the $n$-arrangement are unallocated when $a_i$ becomes the selected agent.

    By the property of envy-free matching with regards the agents, for any $B_j, \forall j\in [k]$, it must be the case that $B_j$ cannot be partitioned into $\kappa_i$ bundles, each with a total item size of at least $\frac{2}{3}$.
    For each $j\in [k]$, by Lemma \ref{lem::cardinal-goods-l-1}, the large and medium items for $a_i$ in the $j$-th column of the $n$-arrangement, denoted as $F_j\cup H_j$, can be placed into at most $\kappa_in$ bundles, each with total item size at most 1.
    Then we claim that 
    $$
    v_i(F_j)+v_i(H_j) + v_i(S_j) < \kappa_i.
    $$
    To see the claim, assume for a contradiction that $v_i(F_j)+v_i(H_j) + v_i(S_j) \geq  \kappa_i$. We next create a $\kappa_i$-partition of $B_j$ such that each bundle has total item size at least $\frac{2}{3}$, contradicting the envy-free matching. 
    First, place $F_j\cup H_j$ into at most $\kappa_i$ bundles, each with a total item size of at most 1. Lemma \ref{lem::cardinal-goods-l-1} guarantees the existence of such a partition.
    Then, if a bundle has total item size less than $\frac{2}{3}$, we add small items such that the total size of the bundle is at least $\frac{2}{3}$.
    Since each small item has size at most $\frac{1}{3}$, after adding small items, the total size of items in the bundle does not exceed 1.
    Therefore, $v_i(F_j)+v_i(H_j) + v_i(S_j) \geq  \kappa_i$ implies that there are enough small items to make each bundle having total item size at least $\frac{2}{3}$, leading to the desired contradiction.

    For $a_i$ and each $\ell \geq k+1$, place the large and medium items in $\ell$-th column of the $n$-arrangement into at most $\kappa_i$ bundles, each with total item size at most 1 (by Lemma \ref{lem::cardinal-goods-l-1}).
    If the number of these bundles is less than $\kappa_i (n - k)$, create empty bundles until the total number reaches $\kappa_i (n - k)$.
    Next, if a bundle has total item size less than $\frac{2}{3}$, we add small items such that the total size of the bundle is at least $\frac{2}{3}$.

    We below show that unallocated small items are enough for producing $\kappa_i(n-k)$ bundles, each with total item size at least $\frac{2}{3}$. 
    For the sake of contradiction, assume that after adding all small items, there exists a bundle with total size less than $\frac{2}{3}$.
    Since each small item has size at most $\frac{1}{3}$ and $a_i$ stops adding small items to a bundle when the total item size becomes at least $\frac{2}{3}$, then for each created bundle created, the total item size is at most 1.
    Therefore, the total size of items in $M\setminus\left(B_1\cup \cdots \cup B_k\right)$ is less than $\kappa_i(n-k)-\frac{1}{3}$.
    Combining the upper bound of total size of items in $B_j$'s, we have that the total size of items for $a_i$ is less than $\kappa_in-\frac{1}{3}$, a contradiction. 
\end{proof}

By Lemmas \ref{lem::cardinal-goods-l-1} and \ref{lem::cardinal-goods-l-2}, at each round of the algorithm, the selected agent can create disjoint parts in Line \ref{step:alg-2/3-properties}. Thus, at each round, there exists at least one matched agent, and therefore, Algorithm \ref{alg:2/3-cardinal-goods} returns $\frac{2}{3}$-CMMS allocations. 
We remark that given the oracle of valuations and MMS partitions, the algorithm runs in time polynomial in $n$ and $m$ by making $n$ partitions queries and $n^3$ valuation queries (for building bipartite graphs).
In fact, we can save the $n^3$ valuation queries by modifying the construction of bipartite graphs.
Whether there exists an edge between $S_j$ and $a_\ell$ can be determined as follows: Let $a_\ell$ fill $\kappa_\ell$ bins with $S_j$ in the same way as that in the proof of Lemma 2, and if each bin has size at least 2/3, then create an edge between $a_\ell$ and $S_j$. 
The above analysis still holds.

To complement the result, we present an example where the existing algorithm, namely Lone Divider, in \citep{DBLP:journals/isci/Aigner-HorevS22} fails to output a $\frac{2}{3}$-CMMS allocation.
The differences between Lone Divider and Algorithm \ref{alg:2/3-cardinal-goods} are: (1) in Line \ref{step:alg-2/3-select-agents}, Lone Divider selects an arbitrary agent, and (2) in Line \ref{step:alg-2/3-properties}, Lone Divider does not impose requirements on the $|\cN|$ parts created by the selected agent.

Consider an instance with 20 agents. Focusing on agent $a_n$ and assuming $\kappa_n=3$, agent $a_n$'s MMS partition has 60 bundles, each composed of three items with size $\frac{2}{3}-\epsilon$, $\frac{1}{3}-\epsilon$, and $2\epsilon$, where $\epsilon>0$ is arbitrarily small.
It is possible that $a_n$ is the last agent and the first 19 allocated bundles are:
12 bundles each with 5 items of size $\frac{2}{3}-\epsilon$, and 7 bundles each with 8 items of size $\frac{1}{3}-\epsilon$.
From $a_n$'s perspective, none of these 19 bundles can satisfy her, as they cannot form 3 bundles each with a total size of at least $\frac{2}{3}$. Thus, this situation can happen in Lone Divider.
There remains 4 items each with size $\frac{1}{3}-\epsilon$ and 60 items each with size $2\epsilon$. It is easy to see that with these remaining items, agent $a_n$ cannot form 3 bundles each with a total size of at least $\frac{2}{3}$.

\section{Ordinal Approximation for the Bin Covering Model}
\label{sec::goods}

In this section, we elaborate on the algorithm that computes $\frac{3}{4}$-OMMS allocations for the bin covering model.  
That is, for any bin covering instance, we can compute an allocation such that each agent $a_i$ can cover at least $\frac{3}{4}\kappa_i-\frac{7}{4}$ bins using her allocated items.

\begin{theorem}
\label{thm:cover:ordinal}
    For any bin covering instance, a $\frac{3}{4}$-OMMS allocation exists. 
\end{theorem} 

For ease of presentation, in this and the next sections, for each agent $a_i$, let $L_i=\{ e_j\in M\mid s_i(e_j)>\frac{1}{2}\}$ be the set of \emph{large} items or $L$-item(s) for short
and $M_i=\{ e_j\in M\mid \frac{1}{2}\geq s_i(e_j)> \frac{1}{3} \}$ be the set of \emph{medium} items or $M$-item(s).
Moreover, we say an item $e$ is \emph{small} or $S$-item for $a_i$ if it is neither large nor medium (i.e., $s_i(e)\leq \frac{1}{3}$). 

\subsection{The Algorithm}
We employ the classic Round-Robin algorithm to compute the desired allocation $\mathbf{A} = (A_1,\ldots,A_n)$. 
That is, the agents sequentially pick items from the unallocated items, each time choosing the one with the largest size.
Since the instances we focus on are IDO, the items allocated to each agent $a_i$ are actually $A_i = \{e_{i+jn}\mid j\in \mathbb{N}, i+jn \le m\}$.
The technical contribution here is to show how each agent $a_i$ distributes her allocated items such that at least $\frac{3}{4}\kappa_i-\frac{7}{4}$ bins can be covered, i.e., $v_i(A_i) \ge \frac{3}{4}\kappa_i - \frac{7}{4}$ for every $a_i \in N$. 
By the characteristic of the Round-Robin algorithm, we have $v_i(A_i) \ge v_i(A_n)$ for every $a_i$. 
Thus, it suffices to prove that $v_i(A_n) \ge \frac{3}{4}\kappa_i - \frac{7}{4}$ holds for every $a_i$.
In the following, we shall focus on a fixed agent $a_i$ and the same arguments apply to the other agents. 

We first construct a set of synthesized items $P_i = \{o_1, \ldots, o_k\}$ (see Algorithm \ref{alg:goods:synthesized}), which consists of all $a_i$'s large items in $A_n$ and synthesizes as many pairs of medium items in $A_n$ as possible. 
The last remaining medium item (which exists when $|M_i\cap A_n|$ is odd) is either directly put into $P_i$ or combined with the smallest synthesized item, depending on the sum of its size and that of the smallest item in $P_i$. 
In the following, instead of directly considering the large and medium items in $A_n$, we will first distribute the synthesized items in $P_i$ and then replace them with their corresponding items in $A_n$. 

\begin{algorithm}[htbp]
\caption{Constructing Synthesized Items for Agent $a_i$}
\label{alg:goods:synthesized}
\begin{algorithmic}[1]
    \REQUIRE Bundle $A_n$, agent $a_i$'s large items $L_i$ and medium items $M_i$. 
    \ENSURE A set of synthesized items $P_i = \{o_1, \ldots, o_k\}$. 
    \STATE Initialized $P_i \leftarrow A_n \cap L_i$. 
    \WHILE{$|A_n \cap M_i| \ge 2$}
        \STATE Arbitrarily pick two items $e_a, e_b$ in $A_n \cap M_i$. 
        \STATE Synthesize $e_a$ and $e_b$ into one item $o$, i.e., $s_i(o) = s_i(e_a) + s_i(e_b)$. 
        \STATE $P_i \leftarrow P_i \cup \{o\}$, $A_n \leftarrow A_n \setminus \{e_a, e_b\}$. 
    \ENDWHILE
    \IF{$|A_n \cap M_i| = 1$}
        \STATE Let $e^*$ be the only item in $A_n \cap M_i$ and $o^*$ be the item with the smallest size in $P_i$. 
        \IF{$s_i(e^*) + s_i(o^*) \ge 1$}
            \STATE $P_i \leftarrow P_i \cup \{e^*\}$. 
        \ELSE
            \STATE Synthesize $e^*$ and $o^*$ into one item $o^\prime$, i.e., $s_i(o^\prime) = s_i(e^*) + s_i(o^*)$.
            \STATE $P_i \leftarrow P_i \setminus \{o^*\} \cup \{o^\prime\}$. 
        \ENDIF
    \ENDIF
    \STATE Sort the items in $P_i$ in decreasing order of the sizes and rename them as $o_1, \ldots, o_k$. 
    \RETURN $P_i$. 
\end{algorithmic}
\end{algorithm}

We then present how to cover bins using the items in $A_n$. 
We consider two cases with respect to the number of synthesized items in $P_i$. 

\smallskip
\noindent\textbf{\underline{Case 1}}: $k \le \frac{3}{4} \kappa_i - \frac{1}{2}$
\smallskip

For this case, we distribute the items in two steps: 
\begin{itemize}
    \item First, we distribute each of the synthesized items in $P_i$ to one bin and cover these bins using the small items in $A_n$. 
    \item Second, we greedily create new bins and cover them using the small items in $A_n$. 
\end{itemize}

\noindent\textbf{\underline{Case 2}}: $k > \frac{3}{4} \kappa_i - \frac{1}{2}$
\smallskip

For this case, we let $z = \lceil\frac{3}{4}\kappa_i - \frac{7}{4}\rceil$ and distribute the items in four steps: 
\begin{itemize}
    \item  First, we cover $k - z$ bins using the smallest $2\cdot (k - z)$ items in $P_i$. 
    \item Second, as long as there remain two items in $P_i$ whose sizes are at most $\frac{2}{3}$, we distribute the smallest two of them to cover a new bin. 
    \item Third, we distribute each of the remaining items in $P_i$ to one bin and cover these bins using the small items in $A_n$. 
    Note that here we specify the order of the bins to be covered, that is, in increasing order of the sizes of the $P_i$ items that they already have, starting from the one with the second smallest $P_i$ item. 
    We lastly cover the bin with the smallest $P_i$ item. 
    \item Fourth, we greedily create new bins and cover them using the small items in $A_n$. 
\end{itemize}

\subsection{The Proof of Theorem \ref{thm:cover:ordinal}}
We are now ready to prove Theorem \ref{thm:cover:ordinal}. 

\begin{proof}[Proof of Theorem \ref{thm:cover:ordinal}]
We will show that the above algorithm can cover at least $\frac{3}{4}\kappa_i - \frac{7}{4}$ bins for $a_i$ in both of the above two cases. 

\smallskip
\noindent\textbf{\underline{Case 1}}: $k \le \frac{3}{4} \kappa_i - \frac{1}{2}$. 

For this case, we first have the following claim. 
\begin{claim}
\label{clm:goods:case1:enough_small}
    There are enough small items in $A_n$ to cover the bins in the first step. 
\end{claim}
\begin{proof}[Proof of Claim~\ref{clm:goods:case1:enough_small}]
    For the sake of contradiction, suppose that we cannot cover all the bins in the first step. 
    Since the size of each small item is at most $\frac{1}{3}$ and the item that cover a bin in the first step is small, the total size that a covered bin in the first step is at most $\frac{4}{3}$. 
    Thus we have an upper bound of the total size of $A_n$ to $a_i$, i.e., $ s_i(A_n) < \frac{4}{3} \cdot (k-1) + 1 \le \kappa_i - 1$.
    However, by the design of the Round-Robin algorithm, we have the following lower bound of the total size of $A_n$ to $a_i$, 
    \begin{equation*}
        s_i(A_n) \ge \frac{1}{n} \cdot (s_i(M) - (n-1)) \ge \kappa_i - 1 + \frac{1}{n}, 
    \end{equation*}
    where the last inequality is because $s_i(M) \ge n \cdot \kappa_i$ by the definition of MMS, a contradiction. 
\end{proof}

By Claim \ref{clm:goods:case1:enough_small}, each bin in the first step is covered and has a total size of at most $\frac{4}{3}$. 
Since all items distributed in the second step are small, each bin in the second step also has a total size of at most $\frac{4}{3}$. 
Thus, letting $ALG$ be the number of covered bins in the above two steps, we have 
\begin{equation}
\label{eq:cover:ordinal:case1:ALG}
    \frac{4}{3} ALG + 1 > s_i(A_n). 
\end{equation}
Besides, by the design of the Round-Robin algorithm, we have the following lower bound of $s_i(A_n)$, 
\begin{equation}
\label{eq:cover:ordinal:case1:sAn}
    s_i(A_n) \ge \frac{1}{n} \cdot (s_i(M) - (n-1)) \ge \kappa_i - 1 + \frac{1}{n}, 
\end{equation}
where the last inequality is because the items $M$ can cover at least $n \cdot \kappa_i$ bins in $a_i$'s MMS partition. 
Combining Inequalities (\ref{eq:cover:ordinal:case1:ALG}) and (\ref{eq:cover:ordinal:case1:sAn}), we have 
\begin{equation*}
    v_i(A_n) \ge ALG > \frac{3}{4}\kappa_i - \frac{3}{2} > \frac{3}{4}\kappa_i - \frac{7}{4}, 
\end{equation*}
which completes the proof of Theorem \ref{thm:cover:ordinal} for Case 1. 

\smallskip
\noindent\textbf{\underline{Case 2}}: $k > \frac{3}{4}\kappa_i - \frac{1}{2}$
\smallskip

For this case, we define more notations here. 
Let $l$ be the number of undistributed items in $P_i$ after the first two steps.
Note that these $l$ items are the $l$ largest ones in $P_i$ (i.e., $o_1, \ldots, o_l$). 
For each $j \in [l]$, let $\Delta_j = 1 - s_i(o_j)$. 
Note that $\Delta_1 \le \dots \le \Delta_l$, $\Delta_{l-1} < \frac{1}{3}$,  and $\Delta_l$ may be at least $\frac{1}{3}$ but at most $\frac{2}{3}$. 
Then we have an upper bound of the total size of the items in $A_n$ that are large or medium for agent $a_i$. 
\begin{claim} 
\label{clm:goods:case2:large_medium_An}$s_i(A_n \cap (L_i \cup M_i)) \le \sum_{j=1}^{l-1}(1-\Delta_j) + (k-l+1) \cdot (1-\Delta_l)$.
\end{claim}
\begin{proof}[Proof of Claim~\ref{clm:goods:case2:large_medium_An}]
    Recall the process of constructing the set $P_i$ and we know that the total size of the items in $P_i$ equals to the total size of the large and medium items in $A_n$. 
    We can derive the upper bound of the total size of $P_i$ as the right side of the inequality by scaling the sizes of $o_{l+1}, \ldots, o_k$ to that of $o_l$ (i.e., $1 - \Delta_l$) and summing up these scaled sizes and the sizes of $o_1, \ldots, o_l$. 
\end{proof}

We further consider two subcases. 

\smallskip
\noindent\underline{\textbf{Subcase 2-1}}: No bin is created in the second step
\smallskip

For this case, there are $2z - k$ undistributed items in $P_i$ after the first two steps, i.e., $l = 2z - k$. 
We are now going to show that there are enough small items in $A_n$ to cover the bins in the third step. 
If it is, together with the $k - z$ covered bins in the first step, we have at least $z = \lceil\frac{3}{4}\kappa_i - \frac{7}{4}\rceil$ covered bins in total, which will prove Theorem \ref{thm:cover:ordinal} for Subcase 2-1.  

We first provide a lower bound of the total size of the small items in $A_n$ to agent $a_i$. 

\begin{claim}
\label{clm:goods:case2:small}
    $s_i(A_n \cap S_i) \ge \sum_{j=1}^{l-1} \Delta_j + (k - l + 1) \cdot \Delta_l + \kappa_i - k - \frac{4}{3} + \frac{4}{3n}$.
\end{claim}
\begin{proof}[Proof of Claim~\ref{clm:goods:case2:small}]
    We first derive an upper bound of the total size of the items in $M$ that are large or medium for $a_i$. 
    Recall that $A_n = \{e_{jn}\mid j\in \mathbb{N}, jn \le m\}$. 
    For each large or medium item $e_{jn}$ in $A_n$, let $Q_j$ be the set of the items between $e_{jn}$ and $e_{(j+1)n}$ (or $e_m$ if $(j+1)n > m$), i.e., $Q_j = \{e_{jn+t} \mid 1 \le t \le n-1, jn+t \le m \}$. 
    Notice that the large and medium items in $A_n$ and the items in their corresponding sets $Q$s cover all the items in $M$ that are large or medium for $a_i$ except the first $n-1$ largest items. 
    By the characteristic of the Round-Robin algorithm, we know that each item in $Q_j$ has a size of at most $e_{jn}$. 
    Therefore, we have that 
    \begin{align*}
        s_i(L_i \cup M_i) &\le n-1 + n\cdot s_i(A_n \cap (L_i \cup M_i)) \\
        &\le n - 1 + n\cdot\sum_{j=1}^{l-1}(1-\Delta_j) + n\cdot(k-l+1) \cdot (1-\Delta_l), 
    \end{align*}
    where $n-1$ is the upper bound of the total size of the first $n-1$ largest items and the second inequality is due to Claim \ref{clm:goods:case2:large_medium_An}. 
    Then we have a lower bound of total size of the items in $M$ that are small for $a_i$, 
    \begin{align*}
        s_i(S_i) = s_i(M) - s_i(L_i \cup M_i) &\ge n\kappa_i - n + 1 - n\cdot\sum_{j=1}^{l-1}(1-\Delta_j) - n\cdot(k-l+1) \cdot (1-\Delta_l) \\
        &= n\cdot\sum_{j=1}^{l-1}\Delta_j + n\cdot(k-l+1) \cdot \Delta_l + (\kappa_i - k)n + 1 - n.
    \end{align*}
    By the Round-Robin algorithm, we have that 
    \begin{align*}
        s_i(A_n \cap S_i) \ge \frac{1}{n}\cdot (s_i(M \cap S_i) - \frac{1}{3}(n-1)) \ge \sum_{j=1}^{l-1} \Delta_j + (k - l + 1) \cdot \Delta_l + \kappa_i - k - \frac{4}{3} + \frac{4}{3n}, 
    \end{align*}
    which completes the proof of the claim. 
\end{proof}

Next, consider when the largest small item in $A_n$ has a size smaller than $\Delta_{l-1}$. 
In order to cover a bin with size $1-\Delta$, the total size of small items we need does not exceed $\Delta+\Delta_{l-1}$. 
Therefore, the total size of small items we need to cover the bins in the third step does not exceed
\[
    \sum_{j=1}^l{(\Delta _j + \Delta_{l-1})} = \sum_{j=1}^{l-1}\Delta_j + l\cdot \Delta_{l-1} + \Delta_l. 
\]
By Claim \ref{clm:goods:case2:small}, it suffices to show that 
\begin{align*}
    \sum_{j=1}^{l-1} \Delta_j + (k - l + 1) \cdot \Delta_l + \kappa_i - k - \frac{4}{3} 
    \ge \sum_{j=1}^{l-1}\Delta_j + l\cdot \Delta_{l-1} + \Delta_l, 
\end{align*}
which is equivalent to $l\cdot \Delta_{l-1} - (k - l) \cdot \Delta_l \le \kappa_i - k - \frac{4}{3}$. 
Since $\Delta_l \ge \Delta_{l-1}$ and $k \ge l$, it suffices to show 
\begin{align*}
    (2l - k)\cdot \Delta_{l-1} \le \kappa_i - k - \frac{4}{3} \Leftrightarrow  \Delta_{l-1} \le \frac{\kappa_i - k - \frac{4}{3}}{4z - 3k}. 
\end{align*}
Since $\Delta_{l-1} < \frac{1}{3}$, it suffices to show 
\begin{align*}
    \Delta_{l-1} \le \frac{\kappa_i - k - \frac{4}{3}}{4z - 3k} \ge \frac{1}{3} \Leftrightarrow z = \lceil\frac{3}{4}\kappa_i - \frac{7}{4}\rceil \le \frac{3}{4}\kappa_i - 1, 
\end{align*}
which is true since $\kappa_i$ is an integer. 

Now, consider when the largest small item in $A_n$ has a size at least $\Delta_{l-1}$.
More generally, let $t$ be the largest integer in $[l-1]$ such that the $j$-th largest small item in $A_n$ has a size at least $\Delta_{l-j}$ for every $j \in [t]$. 
This means that for every $j \in [t]$, the $j$-th largest small item alone is enough to cover the bin with item $o_{l-j}$. 
Since the sizes of the other small items are smaller than $\Delta_{l-t-1}$, the total size of small items we need to cover other bins that have already had a size of $1-\Delta$ does not exceed $\Delta + \Delta_{l-t-1}$. 
Therefore, in order to cover the bins in the third step, the total size of small items we need does not exceed 
\begin{align*}
    \frac{t}{3} + \sum_{j=1}^{l-t-1}(\Delta _j + \Delta_{l-t-1}) + \Delta_l + \Delta_{l-t-1}  =\frac{t}{3} + \sum_{j=1}^{l-t-1}\Delta _j 
    + (l-t)\cdot \Delta_{l-t-1} + \Delta_l, 
\end{align*}
where $\frac{t}{3}$ is the upper bound of the total size of the first $t$ largest small items in $A_n$. 
By Claim \ref{clm:goods:case2:small}, it suffices to show that 
\begin{align*}
    \sum_{j=1}^{l-1} \Delta_j + (k - l + 1) \cdot \Delta_l + \kappa_i - k - \frac{4}{3} \ge \frac{t}{3} + \sum_{j=1}^{l-t-1}\Delta _j + (l-t)\cdot \Delta_{l-t-1} + \Delta_l, 
\end{align*}
which is equivalent to 
\[
    (l-t)\cdot \Delta_{l-t-1} - \sum_{j=l-t}^{l-1}\Delta _j - (k-l) \cdot \Delta_l \le  \kappa_i - k - \frac{4}{3} - \frac{t}{3}. 
\]
Since $\Delta_{l-t-1} \le \cdots \le \Delta_l$, it suffices to show that
\begin{align*}
    (l - t) \cdot \Delta_{l-t-1} - t \cdot \Delta_{l-t-1} - (k-l) \cdot \Delta_{l-t-1} \le \kappa_i - k - \frac{4}{3} - \frac{t}{3},
\end{align*}
which is equivalent to 
\[
\Delta_{l-t-1} \le \frac{\kappa_i - k - \frac{4}{3} - \frac{t}{3}}{4z - 3k - 2t},
\]
which is true since $\Delta_{l-t-1} < \frac{1}{3}$ and $\frac{\kappa_i - k - \frac{4}{3} - \frac{t}{3}}{4z - 3k - 2t} > \frac{1}{3}$.

\smallskip

\noindent\underline{\textbf{Subcase 2-2}}: Some bins are created in the second step. 
\smallskip
For this case, first observe that there are enough small items in $A_n$ such that all the bins in the third step can be satisfied. 
This is because compared to Subcase 2-1, there are now less undistributed items in $P_i$ after the first two steps and the lower bound of the total size of $a_i$'s small items in $A_n$ does not change (i.e., Claim \ref{clm:goods:case2:small} still applies here). 
Also observe that the size of each item in $P_i$ that is distributed to a bin in the first two steps is at most $\frac{2}{3}$. 
This means that each covered bin in the first two steps has a size of at most $\frac{4}{3}$. 
Note that the total size of the items distributed to each covered in the last two steps is obviously at most $\frac{4}{3}$ since the last item given to each bin is a small one. 
Following the same arguments in Case 1, we conclude that there are at least $\frac{3}{4} \kappa_i - \frac{7}{4}$ covered bins. 
\end{proof}
We remark that although the analysis uses MMS partitions and valuations, the algorithm does not require any oracle and runs in polynomial time with respect to $n$ and $m$.

\section{Ordinal Approximation for the Bin Packing Model}\label{sec::chore}

In this section, we consider the ordinal approximation for the bin packing model and present the existence of $\frac{4}{3}$-OMMS allocations.

\begin{theorem}\label{thm::4/3-chore}
    For any bin packing instance, a $\frac{4}{3}$-OMMS allocation exists.
\end{theorem}

\subsection{The Algorithm}
We now informally introduce the algorithm. First, hypothetically arrange all items into $n$ columns (also called $n$-arrangement) based on the descending order of size, such that the $j$-th column contains items $e_j,e_{j+n},e_{j+2n},\ldots$ for every $j\in [n]$.
The algorithm proceeds in $n$ rounds. In each round $j$, we focus on the items in $j$-th column and let a selected agent leave with a bag of items, which is determined through two phases: bag initialization and bag filling. 
In the bag initialization phase (Line \ref{step:chore:3/2:1}), based on the ascending order of item indices, items from the underlying column are added to the bag as long as they are large for at least one remaining agent.
In the bag filling phase (Line \ref{step:chore:3/2:2}), if there exist agents who believe that the total size of items in the bag is no greater than $\frac{1}{n}$ fraction of the total size of all items and still have unallocated small items,
let one of them be the selected agent and add the smallest unallocated item into the bag.
At the end of each round, the selected agent leaves with all items in the bag. The formal description is provided in Algorithm \ref{alg:bin-packing-4/3}.
Our algorithm borrows the idea of the algorithm in \citet{DBLP:conf/nips/0037WZ23}, but we classify large and medium items differently, which necessitates different approaches in the analysis.

\begin{algorithm}[ht!]
 \caption{Computing $\frac{4}{3}$-OMMS allocations}
 \label{alg:bin-packing-4/3}
 \begin{algorithmic}[1]
  \REQUIRE A bin packing instance $\cI=(N,M,\bs)$.
  \ENSURE Allocation $\mathbf{A} = (A_1,\ldots,A_n)$.
        \STATE Initialize $R\gets M$.
        \FOR{$j = 1, \ldots, n$} \label{step:chore:4/3-for}
        \STATE Initialize $B_j\gets \emptyset$, $t\gets j$, $a_k\gets \emptyset$.
        \WHILE{$e_t\in R$ and there exists agent $a_i$ such that $e_t \in L_i\cup M_i$}\label{step:chore:3/2:1}
        \STATE $B_j\gets B_j\cup \left\{ e_t \right\}$, $R\gets R\setminus \left\{ e_t \right\}$, $t\gets t+n$.
        \STATE $a_k\gets a_i$.
        \ENDWHILE
        \WHILE{there exists agent $a_i$ such that $s_i(B_j)\leq \frac{1}{n}s_i(M)$ and there still exist unallocated $S$-items for agent $a_i$}\label{step:chore:3/2:2}
        \STATE $a_k\gets a_i$.
        \STATE Let $e\in R$ be the item with the largest index and $B_j\gets B_j\cup \{ e \}$, $R\gets R\setminus \left\{ e \right\}$.
        \ENDWHILE
        \STATE $A_k\gets B_j$, $ N\gets N\setminus \{ a_k \}$.
        \ENDFOR
 \end{algorithmic}
\end{algorithm}

\subsection{The Proof of Theorem \ref{thm::4/3-chore}}
We will first prove that all items can be allocated. Before that, let us present several propositions regarding the Algorithm \ref{alg:bin-packing-4/3}.
\begin{proposition}\label{prop:ak-add1}
    For an agent $a_i$, suppose that $e^*$ is the last item added to $A_i$. If $e^*$ is small for $i$, then $s_i(A_i\setminus \{ e^* \}) \leq \frac{1}{n}\cdot s_i(M)$.
\end{proposition}
\begin{proof}
    Since $e^*$ is small for $a_i$, then $e^*$ must be added to $A_i$ in the bag filling phase.
    By Line \ref{step:chore:3/2:2}, agent $a_i$ has total size no greater than $\frac{1}{n}s_i(M)$ right before adding $e^*$, and hence, $s_i(A_i\setminus \{ e^* \}) \leq \frac{1}{n}\cdot s_i(M)$.
\end{proof}

\begin{proposition}\label{prop:ak-add2}
    For agents $a_i$ and $a_j$, suppose that $a_j$ receives her bundle $A_j$ in some round earlier than $a_i$. If after the round in which $a_j$ receives $A_j$, there remains unallocated $S$-items for $a_i$, then $s_i(A_j) > \frac{1}{n}s_i(M)$.
\end{proposition}
These propositions directly follow from the condition in Line \ref{step:chore:3/2:2} of Algorithm \ref{alg:bin-packing-4/3}.

\begin{lemma}\label{lem::ak-add1}
    All items are allocated at the termination of Algorithm \ref{alg:bin-packing-4/3}.
\end{lemma}
\begin{proof}
    Suppose that $a_n$ is the last agent to receive a bundle. 
    If at the beginning of round $n$ (Line \ref{step:chore:4/3-for}), all unallocated items are large or medium for $a_n$, then all of them will be added to the bag in Line \ref{step:chore:3/2:1} and then will be allocated to $a_n$ after the bag filling phase.
    Focus on the case when, at the beginning of round $n$,
    there exist unallocated $S$-items for $a_n$. By Proposition \ref{prop:ak-add2}, it holds that for each $j\in [n-1]$, $s_n(A_j)>\frac{1}{n} s_n(M)$, which implies that the total size of remaining items for $a_n$ at the beginning of round $n$ is less than $\frac{1}{n} s_n(M)$.
    Therefore, by Line \ref{step:chore:3/2:2}, all items will be added to the bag and then allocated to $a_n$.
\end{proof}

Next we present the performance guarantee of the returned allocation. Focus on agent $a_i$ and define $H_i:=L_i\cup M_i$. Suppose that $c_i(H_i)=|H_i|-k$. 
Then we claim that the $2k$ items with the smallest size in $H_i$ can be packed into $k$ bins as follows: for each $\ell \in [k]$, pack the item with the $\ell$-th smallest size and the item with the $(2k+1-\ell)$-th smallest size into a bin.
To see the claim, we begin with a solution achieving $c_i(H_i)$ and exchange items as follows: if there exist $e,e'$ such that $s_i(e)<s_i(e')$ with $e$ packed solely in a bin and $e'$ packed with another item, then exchange their positions in the solution.
Then it is easy to see that the $2k$ items with the smallest size in $H_i$ can be packed into $\kappa_i$ bins.

Let $H^*_i\subseteq H_i$ be the set of $2k$ items with the smallest size in $H_i$, and let $H'_i:=H_i\setminus H^*_i$. We below prove that items in $A_i\cap H_i$ can be packed into $\kappa_i+2$ bins.

\begin{lemma}\label{lem::ak-llllast-1}
    For agent $a_i$, it holds that $c_i(A_i\cap H_i) \leq \kappa_i + 2$.
\end{lemma}
\begin{proof}
    We create a packing solution for $A_i\cap H_i$. For each $e\in A_i\cap H'_i$, it will be solely packed into a bin. 
    For items in $A_i\cap H^*_i$, due to the construction of the $n$-arrangement, one can verify that the second largest size item can be packed into a bin with the $t$-th largest size item; the third largest size item can be packed into a bin with the $(t-1)$-th largest size item; and so on, where $t=|A_i\cap H^*_i|$. 
    Thus, excluding the largest size item in $A_i\cap H^*_i$, we can pack the rest items in $\lceil \frac{|A_i\cap H^*_i|-1}{2} \rceil$ bins, and therefore, $c_i(A_i\cap H^*_i) \leq \lceil \frac{|A_i\cap H^*_i|-1}{2} \rceil + 1$.

    The construction of $n$-arrangement implies that $|A_i \cap H'_i|\leq \lfloor \frac{|H'_i|}{n} \rfloor + 1 $ and $|A_i \cap H^*_i|\leq \lfloor \frac{|H^*_i|}{n} \rfloor + 1 $. Then we have
    $$
    \begin{aligned}
    c_i(A_i\cap H_i) & \leq |A_i\cap H'_i| + \lceil \frac{|A_i\cap H^*_i|-1}{2} \rceil + 1 \\
    & \leq \lfloor \frac{|H'_i|}{n} \rfloor + \lceil \frac{\lfloor \frac{|H^*_i|}{n} \rfloor}{2} \rceil + 2 \\
    & \leq \frac{1}{n}\cdot \left( |H'_i| + \frac{|H^*_i|}{2} \right) + \frac{5}{2}.
    \end{aligned}
    $$
    As $|H'_i| + \frac{|H^*_i|}{2} = c_i(H_i) \leq \kappa_i n$ and $c_i(A_i\cap H_i)$ is an integer, the above inequality implies $c_i(A_i\cap H_i) \leq \kappa_i+2$.
\end{proof}

Now we are ready to prove Theorem \ref{thm::4/3-chore}. 

\begin{proof}[Proof of Theorem \ref{thm::4/3-chore}]
If $A_i$ does not contain $S$-items for $a_i$, then by Lemma \ref{lem::ak-llllast-1}, agent $a_i$ is 1-OMMS. Then we focus on the case when $A_i$ contains $S$-items for $a_i$ and let $e^*$ be the last $S$-item added to $A_i$ in the bag filling phase.

We now pack items in $A_i$ into $\kappa_i+2$ bins, such that the total size of items in a bin is at most $\frac{4}{3}$. Consider the following way of packing
\begin{itemize}
    \item Step 1: Pack large and medium items in $A_i$ into $\kappa_i+2$ bins, each with a total size of at most 1.
    \item Step 2: If there exists a bin with total size no greater than 1, add $S$-items to it until the total size just exceeds 1. Repeat the process until all $S$-items are allocated.
\end{itemize}
Lemma \ref{lem::ak-llllast-1} ensures the packing solution for Step 1. For Step 2, by Proposition \ref{prop:ak-add1}, we have $s_i(A_i\setminus \{e^* \}) \leq \frac{1}{n}s_i(M) = \kappa_i$, implying $s_i(A_i)\leq \kappa_i+\frac{1}{3}$ as $e^*$ is small.
Thus, after allocating all $S$-items, the total size of each bin is at most $\frac{4}{3}$; otherwise, $s_i(A_i)>\frac{4}{3}(\kappa_i+2)$, a contradiction.

Begin with the constructed $\kappa_i+2$ bins. If there exists a bin with total size at most $\frac{2}{3}$, reassign an $S$-item from the bin with total size larger than 1 to it. Repeat the process, and at the end, one of the following happens: (i) each bin has total size at most 1, or (ii) each bin has total size at least $\frac{2}{3}$.
If (i) happens, then the theorem statement holds.
If (ii) happens, we first prove that at the end, there are at most $\kappa-4$ bins, each with total size larger than 1.
Assume for a contradiction. As each bin has total size at least $\frac{2}{3}$, the total size of items in $A_i$ is 
$$
s_i(A_i)>(\kappa_i-3)\times 1+5\times \frac{2}{3} = \kappa_i+\frac{1}{3} \geq s_i(A_i),
$$
deriving a contradiction.
For each bin with total size larger than 1, we remove the largest size $S$-item from the bin, ensuring that the total size is at most 1 after the removal (by Step 2).
Then we pack these removed $S$-items into extra bins, each containing 3 small items. Since each $S$-item has size at most $\frac{1}{3}$, each of these bins has total size no greater than 1.
Given that there are at most $\kappa_i-4$ removed $S$-items, we have
$$
    c_i(A_i) \leq \kappa_i+2 + \lceil \frac{\kappa_i-4}{3} \rceil \leq \frac{4}{3}\kappa_i+\frac{4}{3},
$$
completing the proof.
\end{proof}

The algorithm does not require any oracle and runs in polynomial time with respect to $n$ and $m$.
We remark that the $\frac{4}{3}$-CMMS allocations in \citep{DBLP:conf/icml/LiWX} do not translate to $\frac{4}{3}$-OMMS allocations.
In a $\frac{4}{3}$-CMMS allocation, each agent $i$'s bundle can be packed into $\kappa_i$ bins with capacity $\frac{4}{3}$. 
If in a $\frac{4}{3}$-CMMS allocation, each agent $i$'s bundle can be packed such that: whenever the total size of a bin exceeds 1, there exists a small item in the bin whose removal reduces the total item size to at most 1,
then removing these small items and packing them into $\frac{1}{3}\kappa_i$ bins with capacity 1 yields a $\frac{4}{3}$-OMMS allocation.
The $\frac{4}{3}$-CMMS allocation computed in \citep{DBLP:conf/icml/LiWX} do not possess the property just described, and therefore, cannot be converted into $\frac{4}{3}$-OMMS allocations.

\section{Conclusion}
We investigate the cardinal approximation of MMS fair allocation for groups.
From a technical standpoint, our paper contributes novel analyses for both goods and chores. Our result for chores improves the best-known approximation in \cite{DBLP:conf/nips/0037WZ23}.
On a conceptual level, our work provides a more general fair allocation model with non-atomic agents which generalizes the traditional model with atomic agents.
Although the individuals' valuations can be integrated into a unified utility function of a group, the model uncovers several additional insights.
First, the utility function (i.e., the types of social welfare) of each group can be adjusted for different scenarios. 
Second, the group model provides a new dimension to relax MMS fairness, which does not modify the value of MMS and is natural in the situation when we want to maximize the number of satisfied agents. 
Third, we assume the agents in a group have the same valuation function. 
This assumption is natural since agents in the same group typically encounter similar conditions and resources that shape their preferences.
This assumption also builds the relationship between our model and the atomic model with bin-packing and bin-covering valuation functions.
Looking ahead, a promising direction for future research involves extending our model to accommodate heterogeneous valuations.

\bibliographystyle{named}
\bibliography{wine}

\end{document}